\documentclass[a4paper,UKenglish,autoref]{lipics-v2019}

\usepackage[ruled,vlined,linesnumbered,commentsnumbered]{algorithm2e}
\usepackage{xcolor}
\usepackage{hyperref}

\nolinenumbers
\hideLIPIcs

\title             {Parameterized algorithms for node connectivity augmentation problems}
\titlerunning{Parameterized algorithms for node connectivity augmentation problems}

\author{Zeev Nutov}{The Open University of Israel}{nutov@openu.ac.il}
{https://orcid.org/0000-0002-6629-3243}{}
\authorrunning{Zeev Nutov}
\Copyright{Zeev Nutov}


\EventEditors{}
\EventNoEds{2}
\EventLongTitle{}
\EventShortTitle{}
\EventAcronym{}
\EventYear{2020}
\EventDate{}
\EventLocation{}
\EventLogo{}
\SeriesVolume{}
\ArticleNo{}

\begin{document}

\maketitle



\def\AA  {{\cal A}}
\def\BB  {{\cal B}}
\def\CC  {{\cal C}}
\def\FF  {{\cal F}}
\def\TT  {{\cal T}}

\def\p {\partial}
\def\A {\mathbb{A}}
\def\B {\mathbb{B}}
\def\C {\mathbb{C}}

\def\de    {\delta}

\def\empt {\emptyset}
\def\sem  {\setminus}
\def\subs {\subseteq}

\def\f  {\frac}
\def\h {\hat}


\def\kOCA {{\sc $k$-OCA}}    
\def\kCA    {{\sc $k$-CA}}       
\def\ENCA {{\sc $(2,k)$-CA}} 
\def\SSCDS  {{\sc SS-CDS}}  

\def\CA     {{\sc CA}}    
\def\OCA  {{\sc OCA}} 

\def\IBFC    {{\sc Intersecting Biset Family Cover}}  
\def\CFC    {{\sc Crossing Family Cover}}  





\keywords{node connectivity augmentation, fixed parameter tractability} 

\begin{abstract}
A graph $G$ is $k$-out-connected from its node $s$ if it contains $k$ internally disjoint $sv$-paths to every node $v$;
$G$ is $k$-connected if it is $k$-out-connected from every node. 
In connectivity augmentation problems the goal is to augment a graph $G_0=(V,E_0)$ 
by a minimum costs edge set $J$ such that $G_0 \cup J$ has higher connectivity than $G_0$. 
In the {\sc  $k$-Out-Connectivity Augmentation} ({\kOCA}) problem, 
$G_0$ is $(k-1)$-out-connected from $s$ and $G_0 \cup J$ should be $k$-out-connected from $s$; 
in the {\sc $k$-Connectivity Augmentation} ({\kCA}) problem 
$G_0$ is $(k-1)$-connected and $G_0 \cup J$ should be $k$-connected.
The parameterized complexity status of these problems was open even for $k=3$ and unit costs. 
We will show that {\kOCA} and $3$-{\CA} can be solved in time $9^p \cdot n^{O(1)}$,
where $p$ is the size of an optimal solution.
Our paper is the first that shows fixed parameter tractability 
of a {\em $k$-node-connectivity} augmentation problem with {\em high values of $k$}.  
We will also consider the $(2,k)$-{\sc Connectivity Augmentation} ({\ENCA}) problem 
where $G_0$ is $(k-1)$-edge-connected 
and $G_0 \cup J$ should be both $k$-edge-connected and $2$-connected. 
We will show that this problem can be solved in time $9^p \cdot n^{O(1)}$,
and for unit costs approximated within $1.892$. 
\end{abstract}

\section{Introduction} \label{s:intro}

In network design problems the goal is to find a minimum cost subgraph 
that satisfies given connectivity requirement.
Most of these problems are NP-hard, hence parameterized and approximation algorithms are of interest.
A natural question then is whether the problem is {\bf fixed parameter tractable} 
w.r.t a parameter $p$, namely, if it can be solved in time $f(p) \cdot N^{O(1)}$, where $N$ is the input size;
a related question is what approximation ratio can be achieved within this time bound.
One of the most studied problems is the {\sc Steiner Tree} problem,
where we seek a minimum cost subtree that spans a given set of terminals. 
Already in the 70's, Dreyfus and Wagner \cite{DW} showed that this problem can be solved in time $3^q \cdot n^{O(1)}$,
where $q$ is the number of terminals and $n=|V|$ is the number of nodes in the graph; 
for improvements over this running time see \cite{FKM}. 
The Dreyfus-Wagner algorithm extends to the {\sc Directed Steiner Tree} problem, 
in which the goal is to find a minimum cost directed tree that contains a path from a root node $s$ to every terminal. 

Graphs in this paper are assumed to be undirected and may have parallel edges, unless stated otherwise. 
While there was a large progress in the study of parameterized complexity 
of edge-connectivity problems \cite{GU,MV,BFGM,AMPS, FML},
many papers mention that very little is known about their much harder node-connectivity  counterparts.
We will consider the ``simplest'' type of node-connectivity problems, that however have a rich history, 
when the goal is to increase the node connectivity from $k-1$ to $k$ 
from a given node to other nodes, or between all nodes. 
A graph $G=(V,E)$ is {\bf $k$-out-connected from $s$} 
if it contains $k$ internally disjoint $sv$-paths for every $v \in V \sem \{s\}$,
and $G$ is {\bf $k$-connected} if it is $k$-out-connected from every node and $|V| \geq k+1$. 
In the {\sc $k$-Out-Connected Subgraph} problem the goal is to find a minimum cost $k$-outconnected from $s$ spanning subgraph, 
while in the {\sc $k$-Connected Subgraph} problem the spanning subgraph should be $k$-connected. 
These two problems are trivially fixed parameter tractable w.r.t. an optimal solution size, 
since any feasible solution has at least $kn/2$ edges. 
Therefore, it is reasonable to choose as a parameter the number of {\em non-zero} cost 
edges in an optimal solution, a number that may be between $1$ and $\Theta(kn)$. 
This leads to the augmentation versions of these problems, where the goal is to augment a graph $G_0=(V,E_0)$ (of cost zero)
by a minimum costs edge set $J$ such that $G_0 \cup J$ has larger connectivity than $G_0$. 
In this work, we will consider the problem of increasing the connectivity only by $1$. 
Formally, these augmentation problems are as follows. 

\begin{center} \fbox{\begin{minipage}{0.98\textwidth}
\underline{\sc  $k$-Out-Connectivity Augmentation} ({\kOCA}) \\
{\em Input:} \ \ A $(k-1)$-out-connected from $s$ graph $G_0=(V,E_0)$ and an edge set $E$ with costs $\{c_e:e \in E\}$. \\
{\em Output:} A  minimum cost edge set $J \subs E$ such that $G_0 \cup J$ is $k$-out-connected from $s$. 
\end{minipage}} \end{center}

\begin{center} \fbox{\begin{minipage}{0.98\textwidth}
\underline{\sc  $k$-Connectivity Augmentation} ({\kCA}) \\
{\em Input:} \ \ A $(k-1)$-connected graph $G_0=(V,E_0)$ and an edge set $E$ with costs $\{c_e:e \in E\}$. \\
{\em Output:} A  minimum cost edge set $J \subs E$ such that $G_0 \cup J$ is $k$-connected. 
\end{minipage}} \end{center}
 
These are the optimization versions of these problems.
In the decision versions, we are also given a parameter $p$, and ask whether 
there exists a feasible solution $J \subs E$ of size $|J| \leq p$ and cost $c(J) \leq opt$, 
where $opt$ is the optimal solution cost. 
 
Let us briefly review the parameterized and approximation status of these problems.
The directed version of {\sc $k$-Out-Connected Subgraph} admits a polynomial time algorithm \cite{FT}, 
and this implies approximation ratio $2$ for the undirected version.
One the other hand, the approximability status of {\sc $k$-Connected Subgraph} is somewhat complicated;
the problem admits ratio $\lceil \f{k+1}{2}\rceil$ for $2 \leq k \leq 7$ \cite{KhR,ADNP,DN,KN-kcs}, 
$4+\epsilon$ for any constant $k$ and $\epsilon>0$ \cite{N-kcs4,CV}, 
and $O\left(\log k \log \f{n}{n-k}\right)$ for any $k$ \cite{N-comb}. 
The augmentation version {\kCA} admits better approximation ratios for $k \geq 8$: 
ratio $4$ for $n \geq 3k-5$ and  $O\left(\log \f{n}{n-k}\right)$ for any $k$ \cite{N-CSR, N-comb}.
When $(V,E)$ is a complete graph with unit costs (so any edge can be added by a cost of $1$),  
the problem of augmenting an arbitrary graph $G_0$ to be $k$-connected 
can be solved in time $f(k) \cdot n^{O(1)}$ \cite{JJ},
and for a $(k-1)$-connected $G_0$ admits a polynomial time algorithm \cite{V}.
See also a survey in \cite{N-kcs}.

Let {\sc $k$-Edge-Connectivity Augmentation} be the  edge-connectivity version of {\kCA};
note that it is equivalent to the edge connectivity version of $k$-{\OCA}. 
By the cactus model of Dinitz, Karzanov, and Lomonosov \cite{DKL}, 
for  {\sc $k$-Edge-Connectivity Augmentation}, the case of $k-1$ odd and even is equivalent to the case 
$k-1=1$ (so called {\sc Tree Augmentation} problem) and 
$k-1=2$ (so called {\sc Cactus Augmentation} problem), respectively. 
For approximation algorithms for these two problems see recent papers \cite{CTZ,TZ}.

We will consider one of the most common choices of the parameter $p$ -- an optimal solution size.
Nagamochi \cite{Naga} showed that the problem of augmenting a tree by a min-size edge set 
can be solved in time $2^{O(p \log p)} \cdot n^{O(1)}$. 
Guo and Uhlmann \cite{GU} considered both edge an node connectivity versions of $2$-{CA} with unit costs, 
and showed that they have a kernel of size $O(p^2)$.
Marx and V\'{e}gh  \cite{MV} showed that {\sc $k$-Edge-Connectivity Augmentation} with arbitrary costs can be solved in time $2^{O(p \log p)} \cdot n^{O(1)}$,
and that the problem of increasing the edge connectivity from $0$ to $2$ can also be solved within this time bound.
Basavaraju et al. \cite{BFGM} used a novel reduction to the {\sc Node Weighted Steiner Tree} problem to improve the running time
for {\sc $k$-Edge-Connectivity Augmentation} to $9^p \cdot n^{O(1)}$
(see also \cite{N-TAP} for a very simple algorithm for {\sc Tree Augmentation} with running time $16^p \cdot n^{O(1)}$).
The reduction of  Basavaraju et al. \cite{BFGM} was generalized in \cite{N-2CND} to the problem of covering a crossing set family by edges 
(see definitions below), and also was extended to $2$-{\CA}; 
specifically, the result of \cite{N-2CND} implies that $2$-{\CA} and $2$-{\OCA} 
can also be solved in time $9^p \cdot n^{O(1)}$.
However, the parameterized complexity status of both $3$-{\OCA} and $3$-{\CA} was open even for unit costs \cite{MV,GRRW}.

We will show that the reductions of  \cite{BFGM,N-2CND} can be extended to 
$k$-node-connectivity augmentation problems with {\em high values of $k$}.  
But our reduction is not to the {\sc Node Weighted Steiner Tree} problem, 
but to the {\sc Node Weighted} {\sc Group Steiner Tree} problem.
The later can be reduced to the {\sc Directed Steiner Tree} problem, that can be solved in time $3^q \cdot n^{O(1)}$,
where $q$ is the number of terminals.
Formally, we will prove the following.

\begin{theorem} \label{t:1}
{\kOCA} can be solved in time $9^p \cdot n^{O(1)}$. 
\end{theorem}

Algorithms for {\kOCA} sometimes used as subroutines in algorithm for {\kCA}. 
Auletta et al. \cite{ADNP} showed that $3$-{\CA} can be reduced in polynomial time to $3$-{\OCA}. 
For $k=4,5$, \cite{DN-kcs} shows that if {\kOCA} admits approximation ratio $\rho$ then $k$-{\CA} admits ratio $\rho+1$. 
Thus from Theorem~\ref{t:1} we get the following.

\begin{corollary} \label{c:1}
In time $9^p \cdot n^{O(1)}$, $3$-{\CA} can be solved optimally, 
and $4$-{\CA} and $5$-{\CA} can be $2$-approximated.
\end{corollary}

Corollary~\ref{c:1} resolves the open question of parameterized complexity of $3$-{\CA} \cite{MV,GRRW}, 
by establishing that $3$-{\CA} is fixed parameter tractable.

To prove Theorem~\ref{t:1} we will consider a more general problem of covering a biset family by edges.  
A {\bf biset} is an ordered pair $\A=(A,A^+)$ such that $A \subs A^+$;
$\p\A=A^+ \sem A$ is the {\bf boundary} of $\A$, 
$A^*=V \sem A^+$ is the {\bf co-set} of $\A$, and 
$\A^*=(V \sem A^+, V \sem A)$ is the {\bf co-biset} of $\A$.
A biset $\A$ is {\bf proper} if $A \neq \empt$ and $A^* \neq \empt$.  
The intersection and the union of bisets $\A,\B$ are defined by 
$\A \cap \B=(A \cap B,A^+ \cap B^+)$ and $\A \cup \B=(A \cup B,A^+ \cup B^+)$.
Two bisets $\A,\B$ {\bf intersect} if $A \cap B \neq \empt$, and {\bf cross} if also $A^+ \cup B^+ \neq V$.
A biset family $\FF$ is {\bf intersecting/crossing} if $\A \cap \B, \A \cup \B \in \FF$ whenever $\A,\B \in \FF$ intersect/cross.  
An {\bf edge $e$ covers $\A$} if $e$ goes from $A$ to $A^*$, and an edge set $J$ covers $\FF$ 
if every $\A \in \FF$ is covered by some edge in $J$. 
We will identify a set $A$ with the biset $(A,A)$, 
so when $A=A^+$ for every $\A \in \FF$ we have an ordinary set family, 
and we use for set families a similar terminology. 

Let $G_0=(V,E_0)$ be $(k-1)$-out-connected from $s$.
Let us say that a proper biset $\A$ on $V$ is {\bf tight} if  $A \neq \empt$, $s \in A^*$ and $|\p\A|+d_{G_0}(\A)=k-1$, 
where $d_{G_0}(\A)$ is the number of edges in $G_0$ that cover $\A$. 
For the family $\TT$ of tight bisets the following is known, see \cite{F}:
\begin{itemize}
\item
$G_0 \cup J$ is $k$-connected iff $J$ covers $\TT$ (by Menger's Theorem).
\item
$\TT$ is an intersecting biset family. 
\end{itemize}
Thus the following problem includes $k$-{\OCA}.

\begin{center} \fbox{\begin{minipage}{0.98\textwidth}
\underline{\IBFC} \\
{\em Input:} \ \ A graph $G=(V,E)$ with costs $\{c_e:e \in E\}$, an intersecting biset family $\FF$ on $V$. \\
{\em Output:}  A  minimum cost edge set $J \subs E$ that covers $\FF$.
\end{minipage}} \end{center}
In this problem, $\FF$ may not be given explicitly ($|\FF|$ may be exponential in $|V|$);
instead, we require that some queries on $\FF$ can be answered in polynomial time. 
We will assume that for any edge set $I$ we can find in time polynomial in $n=|V|$
the family $\FF^I_{\min}$ and $\FF^I_{\max}$ of all inclusion minimal and maximal bisets, respectively,  
among the bisetsin $\FF$ not covered by~$I$. 
In {\IBFC} instances that arise from the {\kOCA} problem, this can be done using $n$ max-flow computation.
Under this assumption we will prove: 

\begin{theorem} \label{t:2}
{\IBFC} can be solved in time $9^p \cdot n^{O(1)}$.
\end{theorem}

A node $v$ is a {\bf cut-node} of a graph $G$ if $G \sem \{v\}$ is disconnected. 
In addition, we will consider the following problem, that combines both edge and node 
connectivity augmentation:

\begin{center} \fbox{\begin{minipage}{0.98\textwidth}
\underline{$(2,k)$-{\sc Connectivity Augmentation} ({\ENCA})} \\
{\em Input:} \ \ A $(k-1)$-edge-connected graph $G_0=(V,E_0)$ where $k \geq 2$, $Q \subs V$, and an edge set $E$ on $V$ with costs $\{c_e:e \in E\}$. \\
{\em Output:} A  minimum cost edge set $J \subs E$ such that $G_0 \cup J$ is $k$-edge-connected and has no cut-node in $Q$.
\end{minipage}} \end{center}

We will show that {\ENCA} (resp., with unit costs) can be reduced to the {\sc Node Weighted Steiner Tree} problem (resp., with unit weights)
with the following properties:
\begin{enumerate}[(A)]
\item
The neighbors of every terminal induce a clique.
\item
Every non-terminal has at most $2$ terminal neighbors. 
\item
There are no edges between the terminals. 
\end{enumerate}
{\sc Node Weighted Steiner Tree} can be solved in time $3^q \cdot n^{O(1)}$ by the Dreyfus-Wagner algorithm \cite{DW}. 
Byrka et al. \cite{BGA} showed that unit weight instances with properties (A,B,C) 
admit approximation ratio $1.91$,  and Angelidakis et al. \cite{ADS} improved the ratio to $1.892$.
Thus we get the following. 

\begin{theorem} \label{t:3}
{\ENCA} can be solved in time $9^p \cdot n^{O(1)}$, and in the case of unit costs admits approximation ratio $1.892$.
\end{theorem}

In the {\CFC} problem we are
given a graph $G_0=(V,E)$ with edge costs and a symmetric crossing set family $\FF$ on $V$,
and  seek a minimum cost edge set $J \subs E$ that covers $\FF$. 
Theorem~\ref{t:3} was known for the {\CFC} problem \cite{N-2CND}.
We will show that {\ENCA} generalizes {\CFC}, see Section~\ref{s:3}.

Theorems \ref{t:2} and \ref{t:3} are proved in Sections  \ref{s:2} and \ref{s:3}, respectively. 

\medskip

In the rest of the Introduction we survey some additional related work. 
As was mentioned, min-cost connectivity problems that have solution size $\Omega(n)$ 
are trivially fixed parameter tractable w.r.t. an optimal solution size;
this is why our parameter choice is the number of non-zero cost edges in an optimal solution. 
Several other papers studied the so called ``deletion set'' parameter, i.e., 
the number of edges to be removed from the input graph in order to obtain a minimum cost solution.
For example, Bang-Jensen et al. \cite{BBKM} show that the {\sc $k$-Edge Connected Subgraph} problem is 
fixed parameter tractable for the combined parameter of $k$ and the size of a deletion set. 
Gutin et al. \cite{GRRW} show a similar result for {\sc $k$-Connected Subgraph} with unit costs.

In the more general {\sc Survivable Network Design} ({\sc SND}) problem we are given a graph with edge costs 
and pairwise connectivity requirements $\{r_{uv}:uv \in D\}$ over a set $D \subs V \times V$ of demand pairs. 
The goal is to find a min-cost subgraph that contains $r_{uv}$ internally disjoint paths for all $uv \in D$.
In the edge-connectivity version {\sc EC-SND} the paths are only required to be edge disjoint. 
{\sc EC-SND} admits a $2$-approximation algorithm \cite{Jain},
and can be solved in $2^{O(p \log p)} \cdot n^{O(1)}$ time \cite{FML}, where $p$ is the solution size. 
The status of the node-connectivity version {\sc SND} with $r_{uv} \in \{0,1,2\}$ is similar, see \cite{FJW} and \cite{FML}, respectively. 
On the other hand {\sc SND} parameterized by the solution size is W[1]-hard even when $|D|=2$.
{\sc SND} admits approximation ratio $O(k^3 \log n)$ for arbitrary requirements \cite{CK}, 
$O(k^2)$ for rooted requirements, and $O(k \log k)$ for rooted requirements in $\{0,k\}$ \cite{N-r},
where $k$ is the maximum requirement. 
See also a survey in \cite{N-snd}.
For the current status of {\sc SND} problem on directed graphs, see for example \cite{FM} and \cite{N-snd}.

\section{Covering intersecting biset families (Theorem~\ref{t:2})} \label{s:2}

We will reduce {\IBFC} to the following problem:

\begin{center} \fbox{\begin{minipage}{0.98\textwidth}
\underline{\sc Subset Steiner Connected Dominating Set} ({\SSCDS}) \\
{\em Input:} \ \ A  graph $H=(U,I)$, a set $R \subs U$ of terminals, and node weights $\{w_v:v \in U \sem R\}$.   \\
{\em Output:} A min-size node set $J \subs U \sem R$ such that $H[J]$ is connected and $J$ dominates $R$. 
\end{minipage}} \end{center}

As was observed in \cite{BN,N-2CND}, {\SSCDS} reduces to the {\sc Node Weighted} {\sc Group Steiner Tree} problem:
given a graph with node weights $w_v$ and a collection ${\cal S}$ of subsets (groups)
of the node set, find a min-weight subtree that contains a node from every group. 
Given a {\SSCDS} instance $(H,R,w)$ obtain an equivalent {\sc Node Weighted} {\sc Group Steiner Tree} instance  as follows:
for every $r \in R$, introduce a group $S_r$ that consists of the neighbors of $r$ in $H$, and then remove $r$. 
This problem reduces to the {\sc Directed Steiner Tree} problem with $|R|=|{\cal S}|$ terminals, 
that can be solved in time $3^q \cdot n^{O(1)}$, where $q=|R|=|{\cal S}|$.

In fact, we will consider the rooted version {\sc Rooted} {\SSCDS}, when we are also given a non-terminal root node 
$s \in U \sem R$ and we must have $s \in J$. 
All algorithms that we use, as well as hardness results, are applicable to {\sc Rooted} {\SSCDS}.
In the case when property (A) holds, we get the rooted version of the {\sc Node Weighted Steiner Tree} problem,
and then the algorithms of Dreyfus and Wagner \cite{DW} and of Angelidakis et al. \cite{ADS}
apply on (and were in fact designed for) the rooted version.

Let $\FF$ be an intersecting biset family and $J$ an edge set on $V$.
W.l.o.g. we may assume that there is a root node $s \in V \sem \left(\cup_{\A \in \FF} A^+\right)$.  
For $\A,\B \in \FF$ we say that {\bf $\B$ contains $\A$} and write $\A \subs \B$ if $A \subs B$ and $A^+ \subs B^+$.
An {\bf $\FF$-core} is an inclusion minimal biset in $\FF$. 
Let $\CC=\CC_\FF$ denote the set $\FF$-cores. 
For $\C \in \CC$ let $\FF(\C)=\{\A \in \FF: \C \subs \A\}$ denote 
the set of those members of $\FF$ that contain $\C$. 
Given an edge set $J$ and a node set $A$ we will write $J \subs A$ 
and say that $A$ contains $J$ meaning that the set of endnodes of $J$ is contained in $A$. 

Now we define a certain ``separability relation''  on $J \cup \CC \cup \{s\}$ 
and a ''separability graph'' that represents this relation. 
This follows a proof line of \cite{N-2CND}, where the problem of covering a 
symmetric crossing set family $\FF$ was considered. 
In \cite{N-2CND}, two edges $x,y$ are ``separated'' by a set $A$ 
if one of $x,y$ is contained in $A$ and the other in $V \sem A$,
but it is less clear how to extend this definition to bisets. 

\begin{definition} \label{d:separates}
We say that a biset {\bf $\A$ separates} $X \subs \CC \cup J$ from $Y \subs J \cup \{s\}$ 
if $X \cap Y \cap J=\empt$, every $\FF$-core in $X$ is contained in $\A$,
$J \cap X \subs A^+$, and $J \cap Y \subs V \sem A$. 
{\bf $X,Y$ are $\FF$-separable} if such $\A \in \FF$ exists, and $X,Y$ are {\bf $\FF$-inseparable} otherwise.
The {\bf separability graph} $H=(U,I)$ of $\FF,J$ has node set $U=\CC \cup J \cup \{s\}$ and edge set 
$I=\{xy: x \in \CC \cup J,y \in J \cup \{s\} \mbox{ are } \FF\mbox{-inseparable}\}$.
\end{definition}

Note that $s$ and any $\C \in \CC$ are $\FF$-separable, and that $\CC \cup \{s\}$ is an independent set in the separability graph $H$. 
We need the following technical lemma. 
 
\begin{lemma} \label{l:sep}
If $\A$ separates $X$ from $Z$ and $\B$ separates $Y$ from $Z$, 
then (see Fig.~\ref{f:ibg}(a,b)):
\begin{enumerate}[(i)]
\item
$\A \cup \B$ separates $X \cup Y$ from $Z$.
\item
If $A \cap B= \empt$ then $\A$ separates $X$ from $Y$.
\end{enumerate}
\end{lemma}
\begin{proof}
By the assumption of the lemma we have (see Fig.~\ref{f:ibg}(a)):
\begin{itemize}
\item
Since $\A$ separates $X$ from $Z$, 
every $\FF$-core in $X$ is contained in $\A$, every edge in $X$ is contained in $A^+$, 
and every edge in $Z$ is contained in $V \sem A$. 
\item
Since $\B$ separates $Y$ from $Z$,
every $\FF$-core in $Y$ is contained in $\B$, every edge in $X$ is contained in $A^+$, and every 
edge in $Z$ is contained in $V \sem B$.
\end{itemize}
Consequently, 
every $\FF$-core in $X \cup Y$ is contained in $\A \cup \B$,
every edge in $X \cup Y$ is contained in $A^+ \cup B^+$, 
and every edge in $Z$ is contained in $V \sem (A \cup B)$.
This implies that $\A \cup \B$ separates $X \cup Y$ from $Z$.

Now suppose that $A \cap B= \empt$; see Fig.~\ref{f:ibg}(b).
To see that then $\A$ separates $X$ from $Y$, note that
every $\FF$-core in $X$ is contained in $\A$ and every edge in $X$ is contained in $A^+$ 
(since $\A$ separates $X$ from $Z$),
and that  every edge in $Y$ is contained in $V \sem A$ 
(since $A \cap B=\empt$).
\end{proof}

The following key lemma is the technical part of the reduction. 

\begin{lemma} \label{l:ts}
Let $H=(\CC \cup J \cup \{s\},I)$ be the separability graph 
of an intersecting biset family $\FF$ and an edge set $J$ on $V$. 
Let $\C \in \CC$. Then $J$ covers $\FF(\C)$ iff $H$ has a $\C s$-path. 
(Equivalently: $J$ does not cover some $\C \in \FF(\C)$ iff $H$ has no $\C s$-path.)
\end{lemma}
\begin{proof}
Suppose that $J$ does not cover some $\A \in \FF(\C)$. 
Let $\CC_\A$ be the set of $\FF$-cores contained in $\A$ and $J_\A$ the set of edges in $J$ contained in $A^+$.
Then $H$ has no edge between $\CC_\A \cup J_\A$ and $(\CC \cup J \cup \{s\}) \sem (\CC_\A \cup J_\A)$.
Consequently, $H$ has no $\C s$-path. 

\begin{figure}
\centering 
\includegraphics{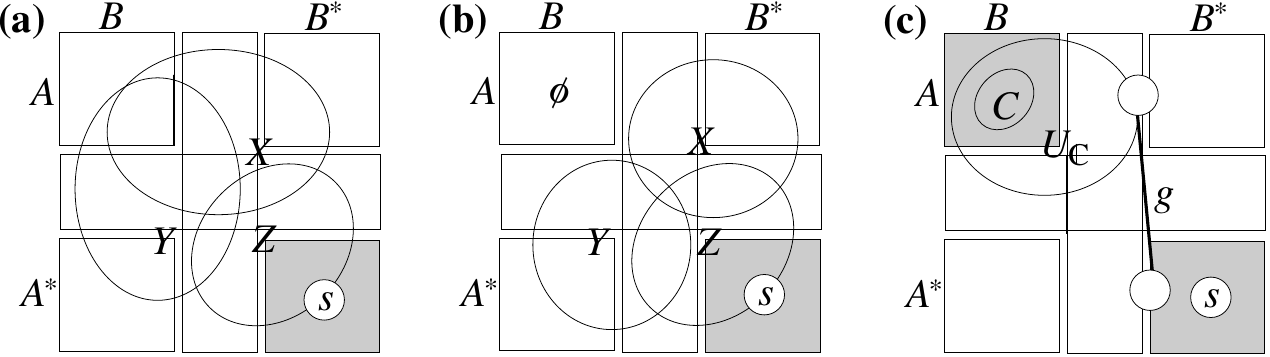}
\caption{Illustration to the proof of Lemmas \ref{l:sep} and \ref{l:ts}.
Dark parts are non-empty, empty parts are marked by $\empt$; 
other parts may or may not be empty.}
\label{f:ibg}
\end{figure}

Suppose that $H$ has no $\C s$-path. 
Let $U_\C$ be the set of nodes of the connected component of $H$ that contains $\C$,
and let $g \in (J \sem U_\C) \cup \{s\}$. 
We now will show that $\FF$ contains a biset that separates $U_\C$ from $g$.
Let $f_0,f_1,\ldots$ be an ordering of $U_\C$, where $f_0=\C$, 
such that in $H$ each $f_i$ with $i \geq 1$ is adjacent to some $f_j$ with $j <i$;
since $H[U_\C]$ is connected, such an ordering exists.
For any $f_i$ there is $\A_i \in \FF$ that separates $f_i$ from $g$.
By Lemma~\ref{l:sep} and since $\FF$ is an intersecting family, for  $X,Y \subs \CC \cup J$ and $Z \subs J \cup \{s\}$,
if $\A$ separates $X$ from $Z$ and $\B$ separates $Y$ from $Z$, 
then (see Fig.~\ref{f:ibg}(a,b)):
\begin{enumerate}[(i)]
\item
$\A \cup \B$ separates $X \cup Y$ from $Z$; moreover, if $\A,\B \in \FF$ and $A \cap B \neq \empt$ then $\A \cup \B \in \FF$. 
\item
If $A \cap B= \empt$ then $\A$ separates $X$ from $Y$.
\end{enumerate}
In particular, if $\A,\B \in \FF$ and $X,Y$ are $\FF$-inseparable, then $A \cap B \neq \empt$ must hold, 
and thus $\A \cup \B$ belongs to $\FF$ and separates $X \cup Y$ from $Z$.
Applying this on $X=\{f_0\},Y=\{f_1\},Z=\{g\}$ and $\A=\A_0,\B=\A_1$,
we get that $\A_1 \cup \A_2 \in \FF$ and separates $\{f_0,f_1\}$ from $g$. 
By an identical argument applied on $X=\{f_0,f_1\}$ and $Y=\{f_2\}$ 
we get that $(A_0 \cup A_1) \cup A_2 \in \FF$ separates $\{f_0,f_1,f_2\}$ from $g$.
By induction we get that $\cup_i A_i \in \FF$ and separates $U_\C$ from $g$.

Let $\A$ be an inclusion minimal biset in $\FF$ that separates $U_\C$ from $s$.
Note that $\A \in \FF(\C)$.
We claim that $J$ does not cover $\A$.
Suppose to the contrary that some $g \in J$ covers $\A$. 
Note that $g \in J \sem U_\C$, thus as we just proved, 
there is $\B \in \FF$ that separates $U_\C$ from $g$, see Fig.~\ref{f:ibg}(c).
Note that $\C \subs \A \cap \B$, hence $\A,\B$ intersect and thus $\A \cap \B \in \FF$.
Summarizing, we have:
\begin{itemize}
\item
$\A$ separates $U_\C$ from $s$ and 
$\B$ separates $U_\C$ from $g$.
\item
$\A \cap \B \neq \empt$ and thus $\A \cap \B \in \FF$.
\end{itemize}
By interchanging the roles of $\A,\B$ and $\A^*,\B^*$ in Lemma~\ref{l:sep},
we get that if $\A$ separates $Z$ from $X$ and $\B$ separates $Z$ from $Y$, then $\A \cap \B$ separates $Z$ from $X \cup Y$. 
Applying this on $Z=U_\C$, $X=\{s\}$, and $Y=\{g\}$, we get that $\A \cap \B$ separates $U_\C$ from $\{s,g\}$.
As $g$ has both ends in $V \sem B$ and covers $\A$, 
it has one end in $A \sem B$ and the other in $A^* \sem B$.
This implies that $\A \cap \B \subsetneq \A$ (namely, $\A \cap \B$ is strictly contained in $\A$). 
Since $\A \cap \B$ separates $U_\C$ from $g$ and $\A \cap \B \in \FF$,
we obtain a contradiction to the minimality of $\A$.  
\end{proof}

From Lemma~\ref{l:ts} we get:

\begin{corollary} \label{c:s}
An edge set $J$ covers an intersecting biset family $\FF$ iff the separability graph $H$ of $\FF,J$ has a subtree that contains $s$ 
and has leaf set $\CC$.  
\end{corollary}

The reduction of {\IBFC} to {\sc Rooted} {\SSCDS} is as follows. 

\begin{definition} \label{d:red}
Given an {\IBFC} instance $(\FF,E,c)$,
the corresponding {\sc Rooted} {\SSCDS} instance $(H,R,w)$ is constructed as follows.
\begin{itemize}
\item
$H$ is the separability graph of $\FF,E$.
\item
$R=\CC$ and the root is $s$.
\item
For every $e \in E$, the weigh of the node $e$ in $H$ equals to the cost of the edge $e$ in $E$. 
\end{itemize}
\end{definition}

From Corollary~\ref{c:s} it follows that $J \subs E$ is a feasible solution to the obtained {\sc Rooted} {\SSCDS} instance
iff $J$ covers $\FF$. By the reduction, the weight of $J$ equals the cost of $J$.   
As was explained at the beginning of this section, {\sc Rooted} {\SSCDS} can be solved in time $3^q \cdot n^{O(1)}$, where $q=|R|$.
Since we need at least ${\sf opt} \geq |\CC_\FF|/2$ edges to cover $\FF$, and since $|R|=|\CC_\FF|+1$,   
we can find an optimal cover of $\FF$ in $3^q \cdot n^{O(1)}=9^p \cdot n^{O(1)}$ time. 

It remains to show that $H$ can be constructed in polynomial time.
For that is is sufficient to show that for any $x \in \CC \cup E$ and $y \in E \cup \{s\}$
we can check in polynomial time whether there is $\A \in \FF$ that separates $x$ from $y$.
Recall that we assume that for any edge set $I$ we can find in polynomial time 
the family $\FF^I_{\min}$ and $\FF^I_{\max}$ of all inclusion minimal and maximal bisets, respectively,  
among the bisets in $\FF$ not covered by $I$. 
From Definition~\ref{d:separates}, it is not hard to verify that 
there is a biset in $\FF$ that separates: 
\begin{itemize}
\item 
$\C \in \CC$ from $e \in E$ iff $e \subs V \sem C$.
\item 
$e \in E$ from $s$ iff there is $\A \in \FF^{\empt}_{\max}$ such that $e \subs A^+$.
\item 
$uv \in E$ from $u'v' \in E$ iff there is $\A \in \FF^I_{\max}$ with $e \subs A^+$ and $e' \subs V \sem A$, where $I=\{u's,v's\}$. 
\end{itemize}

This concludes the proof of Theorem~\ref{t:2}. 

\section{Reduction for  (2,{\em k})-Connectivity Augmentation  (Theorem~\ref{t:3})} \label{s:3}

Let $(G_0=(V,E_0),E,c,Q,k)$ be an instance of {\ENCA} with $k \geq 2$.
Recall that $G_0$ is $(k-1)$-edge-connected, and we seek a min cost edge set $J \subs E$ 
such that $G_0 \cup J$ is both $k$-edge-connected and has no cut node in $Q$.  
We will construct an equivalent instance $(H=(U,F),w,R)$ of {\sc Rooted} {\SSCDS} that satisfies the following two properties:
\begin{enumerate}[(A)]
\item
The neighbors of every terminal induce a clique.
\item
Every non-terminal has at most $2$ terminal neighbors. 
\end{enumerate}
It is easy to see that if property (A) holds, then any subtree of $H$ that contains $R$ 
can be converted into a subtree of the same cost with leaf set $R$.
Thus any {\SSCDS} instance that satisfies property (A)
is equivalent to the {\sc Node Weighted Steiner Tree} instance with the same graph, weights, and set of terminals. 

\begin{definition} \label{d:tight}
Given a $(k-1)$-connected graph $G_0=(V,E_0)$ with $k \geq 2$ and $Q \subs V$, 
we assign capacities to the nodes $q(v)=k-1$ if $v \in Q$ and $q(v)=\infty$ otherwise.
We say that a proper biset $\A$ is {\bf tight} if $d_{G_0}(\A)+q(\p\A)=k-1$ and 
denote by $\TT=\TT_{G_0}$ the family of tight bisets.
\end{definition}

Equivalently, a proper biset $\A$ is tight if either 
$\p\A=\empt$ and $d_{G_0}(\A)=k-1$, or
$\p\A$ is a single node in $Q$ and  $d_{G_0}(\A)=0$.
Namely, the family of tight bisets is a union $\TT=\AA \cup \BB$ of 
a set family $\AA$ and a biset family $\BB$ defined by
$$
\AA=\{A: d_{G_0}(A)=k-1\} \ \ \ \ \ \BB=\{\B: \p\B \mbox{ is a single node in } Q, d_{G_0}(\B)=0, B,B^* \neq \empt\} \ .
$$

Note that since $G_0$ is $(k-1)$-edge-connected, $d_{G_0}(\A)+q(\p\A) \geq k-1$ for every proper biset $\A$.
From Menger's Theorem it follows that $J$ is a feasible {\ENCA} solution iff $J$ covers the family 
$\TT=\{\A: A,A^* \neq \empt, d_{G_0}(\A)+q(\p\A)=k-1\}$ of tight  bisets.

We need some definitions. 
The {\bf co-biset} of a biset $\A$ is the biset $\A^*=(V \sem A^+,V \sem A)$. 
A biset family is symmetric if $\A^* \in \FF$ whenever $\A \in \FF$. 
Two bisets $\A,\B$ {\bf co-cross} if $A \sem B^+,B \sem A^+$ are both non-empty.

\begin{lemma} \label{l:FAB}
The family $\TT$ of tight bisets is symmetric and crossing,
and any $\A,\B \in \TT$ cross or co-cross. 
Consequently, $\C \subs \A$ or $\C \subs \A^*$ holds for any $\A \in \TT$ and a $\TT$-core $\C \in \CC_\TT$.  
\end{lemma}
\begin{proof}
Define a biset function $f(\A)=d_{G_0}(\A)+q(\p\A)$.
Since $f(\A)=f(\A^*)$ for any biset $\A$, we get that $\TT$ is symmetric.  
We will show that $\TT$ is a crossing family. 
The functions $f(\A)$ is submodular, namely,  
$f(\A)+f(\B) \geq f(\A \cap \B)+f(\A \cup \B)$ holds for any two bisets $\A,\B$;
this is so since it is known that each of the functions $d_{G_0}(\A)$ and $q(\p\A)$ is submodular.
If $\A,\B$ cross then $\A \cap \B,\A \cup \B$ are both proper bisets and thus 
$f(\A \cap \B),f(\A \cup \B) \geq k-1$. This implies
\[
k-1+k-1=f(\A)+f(\B) \geq f(\A \cap \B)+f(\A \cup \B) \geq k-1+k-1 \ .
\] 
Thus equality holds everywhere, hence $f(\A \cap \B)=f(\A \cup \B)=k-1$.

\begin{figure}
\centering 
\includegraphics{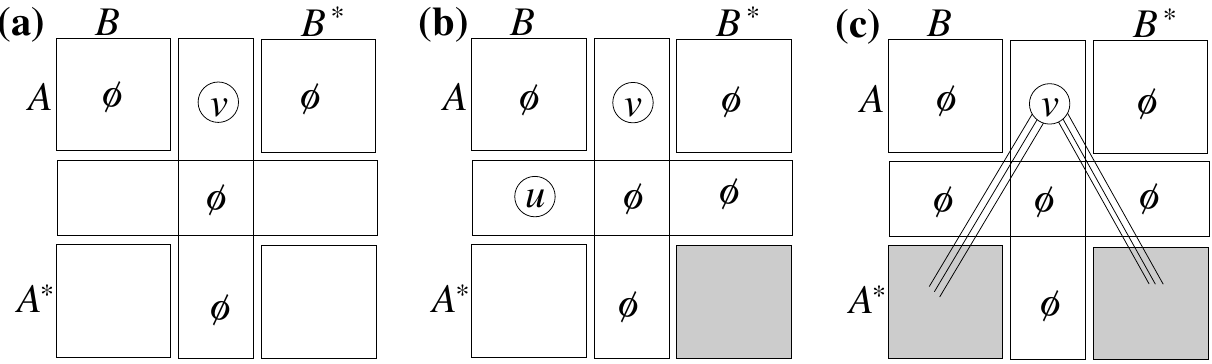}
\caption{Illustration to the proof of Lemma~\ref{l:FAB}. 
Dark parts are non-empty, empty parts are marked by $\empt$; 
other parts may or may not be empty.}
\label{f:ab}
\end{figure}

Now we show that any $\A,\B \in \TT$ cross or co-cross. 
Suppose to the contrary that there are $\A,\B \in \TT$ that do not cross nor co-cross. 
Then $A \subs \p\B$, or $A^* \subs \p\B$, or $B \subs \p\A$, or $B^* \subs \p\A$;
say $A \subs \p\B$; see Fig.~\ref{f:ab}(a). 
Then $A \cap \p\B \neq \empt$ since $A \neq \empt$, and $\p\B=\{v\}$ for some $v \in Q \cap A$ since $\B$ is tight. 

Suppose that $\p\A \neq \empt$, say $\p\A=\{u\}$  for some $u \in Q \cap B$; see Fig.~\ref{f:ab}(b). 
This implies $A^* \cap B^*=A^* \neq \empt$. Since $\A,\B$ are tight, $G_0$ has no edge between $A^* \cap B^*$ 
and $\{u,v\}$, contradicting that $G_0$ is connected.

Suppose that $\p\A = \empt$. 
Then since $\B$ is a proper biset, none of $A^* \cap B,A^* \cap B^*$ is empty; see Fig.~\ref{f:ab}(c). 
Since $\p\B \neq \empt $ and since $\B$ is tight, there is no edge between $B \cap A^*$ and $B^* \cap A^*$. 
This gives the contradiction
$k-1+k-1 \leq d_{G_0}(\B \sem \A)+d_{G_0}(\A \cup \B) = d_{G_0}(\A)=k-1$.
\end{proof}

\begin{lemma} \label{l:int}
Fix some $\TT$-core $\C_0$ and $s \in C_0$,
and let $\FF=\{\A \in \TT: s \in A^*\}$.  
Then $\FF$ is an intersecting biset family 
and $J$ is a feasible {\ENCA} solution iff $J$ covers $\FF$. 
\end{lemma}
\begin{proof}
To prove the lemma it is sufficient to show that if $J$ covers $\FF$ then $J$ covers $\TT$. 
Let $\A \in \TT$. By Lemma~\ref{l:FAB}, $\C_0 \subs \A$ or $\C_0 \subs \A^*$.
If $\C_0 \subs \A$ then $\A^* \in \FF$ and if $\C_0 \subs \A^*$ then $\A \in \FF$. 
Thus $J$ covers $\A$ or $\A^*$, which is equivalent to covering $\A$. 
\end{proof}

Let $\CC=\CC_\FF$ be the family of $\FF$-cores. 
Let $(H=(U,F),R,s,w)$ be a {\sc Rooted} {\SSCDS} instance as in Definition~\ref{d:red}, where 
$H$ is the separability graph of $\FF,E$, $R=\CC$, the root is $s$, and 
for every $e \in E$, the weigh of the node $e$ in $H$ equals to the cost of $e$.
By Lemma~\ref{l:FAB} and Corollary~\ref{c:s}, 
$G_0 \cup J$ is a feasible solution to {\ENCA} iff $H[J]$ is connected and $J$ dominates $R$ in $H$. 
The first part of Theorem~\ref{t:3} now follows from Lemma~\ref{l:int} and Theorem~\ref{t:2}.  

For the second part of Theorem~\ref{t:3} we will prove the following. 

\begin{lemma}
The {\SSCDS} instance $(H=(U,F),w,R)$ satisfies properties (A,B). 
\end{lemma}
\begin{proof}
We prove property (A). Consider an $\FF$-core $\C$. 
Let $uv,xy \in E$ such that each of the pairs $uv,\C$ and $xy,\C$ is $\FF$-inseparable.
Then $\{u,v\} \cap C \neq \empt$ and $\{x,y\} \cap C \neq \empt$; say, $u,x \in C$.
Let $\A \in \FF_s$.
If $\C \subs \A$ then $u,x \in A$.
If $\C \subs \A^*$ then $u,x$ in $A^*$.
In both cases, $\A$ cannot separate one of $uv,xy$ from the other. 

Property (B) follows from the fact that if $\C$ and $uv$ are $\FF$-inseparable then $u \in C$ or $v \in C$, 
and since $C \cap C'=\empt$ holds for any two $\FF$-cores.
\end{proof}

This concludes the proof of Theorem~\ref{t:3}. 

\medskip

Now we will show that {\ENCA} generalizes the {\CFC} problem. 
Let $(\FF,E,c)$ be an instance of {\CFC}, where $\FF$ is a symmetric crossing set family on a groundset $U$, 
and $E$ is an edge set on with costs $\{c_e:e \in E\}$.  
We will show how to construct an equivalent $(2,3)$-{\CA} instance.
A {\bf cactus} is a $2$-edge-connected graph in which any two cycles have at most one node in common 
(equivalently: every block of the graph is a cycle).
Dinitz, Karzanov, and Lomonosov \cite{DKL} showed that the family $\FF$ of
mini\-mum edge cuts of a graph $G$ on node set $U$ can be represented by the family $\TT$ of minimum edge cuts of a cactus $G_0=(V,E_0)$
and a mapping $\psi:U \longrightarrow V$, such that $\FF=\{\psi^{-1}(A):A \in \TT\}$.
Dinitz and Nutov \cite[Theorem 4.2]{DN} (see also \cite[Theorem 2.7]{N-Th}) extended this representation 
by showing that an arbitrary symmetric crossing family $\FF$ can be represented 
by $2$-edge cuts and specified $1$-node cuts of a 
cactus.\footnote{A representation identical to the one of \cite{DN} was announced later by Fleiner and Jord\'{a}n \cite{FJ}.}
This representation can be stated as follows.

\begin{theorem} [\cite{DN}] \label{t:DN}
Let $\FF$ be a crossing family on a groundset $U$. 
Then there exists a cactus $G_0=(V,E_0)$, a mapping $\varphi:U \longrightarrow V$, 
and a set $Q$ of cut-nodes of $G_0$ with $\psi^{-1}(Q)=\empt$, such that 
$\FF=\{\varphi^{-1}(A): A \in \AA\} \cup \{\varphi^{-1}(B): \B \in \BB\}$, where 
$$
\AA=\{A: d_{G_0}(\A)=2\} \ \ \ \ \ \BB=\{\B: \p\A \mbox{ is a single node in } Q, d_{G_0}(\B)=0, B,B^* \neq \empt\} \ .
$$
Furthermore, if for any $A,B \in \FF$ the set $(A \sem B) \cup (B \sem A)$ is not in $\FF$, then $Q=\empt$.
\end{theorem}

Given a {\CFC} instance $(\FF,E,c)$ construct a $(2,3)$-{\CA} instance $(G_0=(V,E_0), E',Q,c')$ as follows.
\begin{itemize}
\item
$G_0$ and $Q$ (and $\varphi$) are as in Theorem~\ref{t:DN}.
\item
For every edge $e=uv$ in $E$ there is an edge $e'=\varphi(u)\varphi(v)$ in $E'$ of cost $c'(e')=c(e)$. 
\end{itemize}
Then $\TT=\AA \cup \BB$ is the family of tight sets of $G_0$, 
and $J \subs E$ covers $\FF$ iff $J' \subs E$ covers $\TT$,
where $J'=\{e':e \in J\}$.
Note that in the obtained  $(2,3)$-{\CA} instance no edge in $E'$ is incident to a node in $Q$,
while in general $(2,3)$-{\CA} instances such edge might exist. 
This suggests that $(2,3)$-{\CA} strictly generalizes the {\CFC} problem. 


\end{document}